\documentclass[11pt]{article}
\usepackage{fullpage,amsmath,amssymb,amsthm}

\newcommand{\R}{\mathbb{R}}
\newcommand{\dist}{\mathrm{dist}}
\renewcommand{\P}{\mathbb{P}}
\newcommand{\E}{\mathbb{E}}
\newcommand{\parent}{\mathrm{parent}}

\newcommand{\vt}{v^{(t)}}
\newcommand{\OSS}{\textsf{Online Vector Sparsification}}

\def\bone{{\mathbf 1}}

\def\to{{\rightarrow}}

\newtheorem{definition}{Definition}[section]

\newtheorem{lemma}[definition]{Lemma}
\newtheorem{theorem}[definition]{Theorem}

\newtheorem{claim}[definition]{Claim}
\newtheorem{prop}[definition]{Proposition}
\newtheorem{remark}{Remark}

\title{\textbf{An Alon-Boppana Type Bound for  Weighted Graphs and Lowerbounds for
Spectral Sparsification}}

\author{Nikhil Srivastava\footnote{{\sf nikhil@math.berkeley.edu.} {U.C.
Berkeley}.  Supported by NSF grant CCF-1553751 and a Sloan research fellowship.} \and Luca Trevisan\footnote{{\sf luca@berkeley.edu.} U.C. Berkeley. This material is based upon work supported by the National Science Foundation under Grants No. 1540685 and No. 1655215.}}

\begin{document}
\maketitle

\begin{abstract}
We prove the following Alon-Boppana type theorem for general (not necessarily regular) weighted graphs: if $G$ is an $n$-node weighted undirected graph of average combinatorial degree $d$ (that is, $G$ has $dn/2$ edges) and girth $g> 2d^{1/8}+1$, and if $\lambda_1 \leq \lambda_2 \leq \cdots \lambda_n$ are the eigenvalues of the (non-normalized) Laplacian of $G$, then 
\[ \frac {\lambda_n}{\lambda_2} \geq 1 + \frac 4{\sqrt d} - O \left( \frac 1{d^{\frac 58} }\right)  \]
(The Alon-Boppana theorem implies that if $G$ is unweighted and $d$-regular,
then $\frac {\lambda_n}{\lambda_2} \geq 1 + \frac 4{\sqrt d} - O\left( \frac 1 d \right)$ if the diameter is at least $d^{1.5}$.)

Our result implies a lower bound for spectral sparsifiers. A graph $H$ is a spectral $\epsilon$-sparsifier of a graph $G$ if
\[ L(G) \preceq L(H) \preceq (1+\epsilon) L(G) \]
where $L(G)$ is the Laplacian matrix of $G$ and $L(H)$ is the Laplacian matrix of $H$.
Batson, Spielman and Srivastava proved that for every $G$ there is an $\epsilon$-sparsifier $H$ of average degree $d$ where $\epsilon \approx \frac {4\sqrt 2}{\sqrt d}$ and the edges of $H$ are a (weighted) subset of the edges of $G$. Batson, Spielman and Srivastava also show that the bound on $\epsilon$ cannot be reduced below $\approx \frac 2{\sqrt d}$ when $G$ is a clique; our Alon-Boppana-type result implies that $\epsilon$ cannot be reduced below $\approx \frac 4{\sqrt d}$ when $G$ comes from a family of expanders of super-constant degree and super-constant girth. 

The method of Batson, Spielman and Srivastava proves a more general result, about sparsifying sums of rank-one matrices, and their method applies to an ``online'' setting. We show that for the online matrix setting the $4\sqrt 2 / \sqrt d$ bound is tight, up to lower order terms.
\end{abstract}

\newpage

\newcommand{\Diam}{\mathsf{diam}(G)}
\section{Introduction}

If $G$ is an (unweighted, undirected) $d$-regular graph on $n$ vertices, and if
$A$ is its adjacency matrix, then the largest eigenvalue of $A$ is $d$, and the
spectral expansion of $A$ is measured by the range of the other eigenvalues: the
smaller the range, the better the expansion. If one denotes by $d = \lambda_1(A)
\geq \lambda_2(A) \geq \cdots \lambda_n(A)$ the eigenvalues of $A$, then
Alon and Boppana \cite{Nilli91} showed that there is a limit to how concentrated
these eigenvalues can be as a function of $d$, namely:
\[ \lambda_2(A) \geq 2 \sqrt{d-1} - O \left( \frac {\sqrt{d}}{\Diam} \right) \]
where $\Diam$ is the diameter of $G$, and:
\[ \lambda_n(A) \leq - 2 \sqrt{d-1} + O \left( \frac {\sqrt{d}}{\Diam} \right). \]
Thus, for every fixed $d$, an infinite family of $d$-regular graphs will satisfy
$\lambda_2(A) \geq 2\sqrt{d-1} - o_n(1)$
and $\lambda_n(A) \leq 2\sqrt{d-1} + o_n(1)$.

Lubotzky, Phillips and  Sarnak \cite{LPS88} call a $d$-regular graph {\em
Ramanujan} if it meets the Alon-Boppana bound: \[ 2 \sqrt{d-1} \geq
\lambda_2(A)
\geq \lambda_n(A) \geq - 2\sqrt{d-1} \ , \] and they show that infinite families of
Ramanujan graphs exist for every degree such that $d-1$ is prime. Friedman
\cite{Friedman08} shows that for every fixed $d$ there is an ``almost
Ramanujan''  family of $d$-regular graphs (one for each possible number of
vertices) such that $\lambda_2(A) \leq 2\sqrt{d-1} +o_n(1)$ and $\lambda_n(A) \geq -
2\sqrt{d-1} - o_n(1)$. Furthermore, infinite families of bipartite Ramanujan
graphs (that is, bipartite graphs such that $\lambda_2 \leq 2 \sqrt{d-1}$) are
known to exist for every degree and every even number of vertices
\cite{MSS13:ram,MSS15} and to be efficiently constructible  \cite{Cohen16}.

Given our precise understanding of the extremal properties of the spectral
expansion of regular  graphs, there has been considerable interest in exploring
generalizations of the above theory to non-regular and/or weighted graphs. There
are at least three possible generalizations which have been considered.

\subsubsection*{Universal Covers}
Ramanujan graphs have the property that the range of their non-trivial adjacency
matrix eigenvalues is bounded by the support of the spectrum of their universal
cover (the infinite $d$-regular tree). Thus, Hoory, Linial and Wigderson \cite
{HLW06} define irregular Ramanujan graphs as graphs whose range of non-trivial
eigenvalues is contained in the spectrum of their universal  cover, or, in the
``one sided version'' as graphs whose second largest eigenvalue is at most the
spectral radius of the universal cover. An ``Alon-Boppana'' bound showing that
in any infinite family of graphs $\lambda_2(A)$ becomes arbitrarily close to the
spectral radius of the universal cover is proved in \cite{G95}. Existence proofs
of infinite families of irregular Ramanujan graphs according  to this definition
are presented in \cite{MSS13:ram}.

\subsubsection*{Normalized Laplacians} Another interesting notion of expansion
for irregular graphs is to require all the non-trivial eigenvalues of the
transition matrix of the random walk on $G$ to be in small range or,
equivalently, to require all the eigenvalues of the normalized Laplacian matrix
$\bar L = I - D^{-1/2} A D^{-1/2}$ to be in a small range around 1. This is a
natural definition because control of the normalized Laplacian eigenvalues
guarantees some of the same properties of regular expanders, such as bounds on
the diameter and a version of the expander mixing lemma. 

For a $d$-regular graph, if $\lambda_i$ is the $i$-th largest eigenvalue of the
adjacency matrix, then $1-\lambda_i/d$ is the $i$-th smallest eigenvalue of the
normalized Laplacian matrix. If we denote by $0 = \lambda_1(\bar L) \leq
\lambda_2(\bar L) \leq \cdots \lambda_n (\bar L)$ the normalized Laplacian eigenvalues of a $d$-regular graph,
the Alon-Boppana bounds become: 

\begin{equation} \label{ab.normalized} \lambda_2(\bar L) \leq 1- 2\frac
{\sqrt{d-1}}d + o_n(1), \ \ \  \lambda_n(\bar L) \geq 1+  2\frac {\sqrt{d-1}}d -
o_n(1) \ .  \end{equation}

It might be natural to conjecture that the above bounds hold also for irregular
graphs, if we let $d$ be the {\em average degree}, thus putting a limit to the
expansion of sparse graphs, regardless of degree sequence. Young \cite{Young11},
however, shows that this is not the case, and he exhibits families of graphs of
average degree $d$ such that  $\lambda_2(\bar L) \leq 1- 2\frac {\sqrt{d-1}}d -
\epsilon$, where $\epsilon> 0$ depends on $d$ but not on the size of the graph.
It would interesting to see if \eqref{ab.normalized} holds for irregular graphs
with an error term $o(1/\sqrt d)$ dependent on $d$.  Young \cite{Young11} and
Chung \cite{Chung16} prove Alon-Boppana type bounds for irregular (unweighted) graphs based
on a parameter that depends on the first two moments of the degree distribution
but is in general incomparable to \eqref{ab.normalized}. 

\subsubsection*{Spectral Sparsifiers} The notion of {\em spectral
sparsification} of graphs can also be seen as a generalization of the notion of
expansion to graphs that are weighted and not necessarily regular. Recall that a
(weighted, not necessarily regular) graph $G$ is called a $(1+\epsilon)$
spectral sparsifier  of $G'$ if $G$ has the same set of vertices and a weighted
subset\footnote{Note that one can consider sparsifiers which use edges outside
$G$. However, in all known constructions and in many applications $G'$ is required to
be a subset of $G$, so we take this as part of the definition, since it is
necessary for our lower bound.}  of the edges of $G$ and 
\[ L(G') \preceq L(G) \preceq (1+\epsilon) \cdot  L(G') \]
where $L(G)$ is the (non-normalized) Laplacian matrix $D - A$ of $G$. This
notion, introduced by Spielman and Teng \cite{ST04}, strengthens the notion of
{\em cut sparsifier} defined by  Bencz\'ur and Karger \cite{BK96}. It can be
seen as a generalization of the notion of expander, because if $K$ is clique and
$G$ is a $(1+\epsilon)$-sparsifier of $K$, then $G$ has several of the useful
properties of expander graphs, and it satisfies a version of the expander mixing
lemma. Since the Laplacian of any clique is a multiple of the identity
orthogonal to the all ones vector, it is
easy to see that $G$ is a $(1+\epsilon)-$spectral sparsifier of a clique if and
only if
\begin{equation}\label{eq.relative}\frac {\lambda_n(L(G))}{\lambda_2(L(G))} \leq 1 +
\epsilon.\end{equation}
Thus, another notion of expansion for irregular weighted graph is to
consider the {\em relative} range of the non-trivial eigenvalues of the
unnormalized Laplacian matrix.

Batson, Spielman and Srivastava \cite{BSS12} showed
that for every $G'$ there is a $(1+\epsilon)$ sparsifier $G$ of average degree
$d$ (i.e., $dn/2$ edges) such that $\epsilon \leq \frac {4\sqrt 2}{\sqrt d} + O \left( \frac 1d
\right)$. 
However, their work left a gap in our understanding of the precise dependence of
$\epsilon$ on $d$: they proved that it is not possible to do better than
$\epsilon\approx \frac{2}{\sqrt{d}}$ and conjectured that this could be improved
to $\frac{4}{\sqrt{d}}$. The number $1+\frac{4}{\sqrt{d}}$  corresponds to the ``Ramanujan'' bound
obtained by approximating the complete graph by a Ramanujan graph $R_d$, since
for such a graph we have:

\[ \frac {\lambda_n(L(R_d))}{\lambda_2(L(R_d)) } \leq \frac {d + 2\sqrt{d-1}}{d-2\sqrt{d-1}} 
\leq 1 + \frac 4{\sqrt d} + O \left( \frac 1d\right), \]
which is also best possible for unweighted regular graphs up to $o_n(1)$ terms by the Alon-Boppana bound.

Thus, \cite{BSS12} called their construction a ``twice Ramanujan
sparsifier'' because, when applied to a clique, it has twice the number of edges
($dn$ instead of $dn/2$)
of a Ramanujan graph for the same $(1+\frac4{\sqrt{d}})$-approximation. Equivalently,
if one applies their construction to create a $(1+\epsilon)$-sparsifier of the
clique of average combinatorial\footnote{For weighted graphs, the term
``degree'' can be ambiguous, so from this point forward we will call the number
of edges incident on a vertex the {\em combinatorial} of the vertex, and we will
call the total weight of the edges incident on a vertex the {\em weighted}
degree of the vertex.}  degree $d$, then one obtains $\epsilon$ that is a
factor of $\sqrt 2$ off from what would have been possible using a true
$d-$regular Ramanujan graph.

\subsection{Our Results} 

\subsubsection{An Alon-Boppana-type Bound on $\lambda_n/\lambda_2$}

Our work clarifies the dependence on $d$ in the Spectral Sparsification context
described above.
We prove the following Alon-Boppana type lower bound on  $\lambda_n  /
\lambda_2$ on the Laplacian matrices of weighted graphs with moderately large girth.

\begin{theorem} \label{th.ab} Let $G$ be a weighted unidrected graph with $n$ vertices
and $dn/2$ edges. Let $\lambda_1 \leq \cdots \le \lambda_n$ be the eigenvalues of the non-normalized Laplacian matrix of $G$. If
the girth of $G$ is at least $2d^{1/8}+1$, then
\[ \frac {\lambda_n}{\lambda_2} \geq 1 + \frac 4{\sqrt d} - O \left( \frac 1 {d^{5/8}} \right) - O \left( \frac 1n \right)  \]
\end{theorem}

This result shows that the dependence of $\epsilon$ on $d$ in spectral
sparsification cannot be better than $1+\frac{4}{\sqrt{d}}$ up to lower order
terms in $d$, as follows. Let $G'_n$ be a family of $D_n$-regular graphs such
that all the non-trivial Laplacian eigenvalues are in the range $D_n\cdot (1\pm
o_n(1))$ and with girth going to infinity (the LPS expanders \cite{LPS88} have
this property). Then any $(1+\epsilon)$-spectral sparsifier $G_n$ of $G'_n$ of average
degree $d$ must have girth greater than $d^{1/8}$ for sufficiently large $n$, so
our theorem implies that $\epsilon \geq 4/\sqrt d - O(1/d^{5/8})$, whence $G_n$
cannot be a better than $(1+4/\sqrt{d}-o_n(1)-o(1/\sqrt{d}))$-sparsifier of
$G'_n$. This improved bound implies that the ``Ramanujan'' quality approximation remains optimal
in the broader category of weighted graphs --- previously \cite{BSS12}, it was conceivable that it is somehow
possible to achieve $1+2/\sqrt{d}$ using variable weights. 

Our proof of Theorem \ref{th.ab} involves the construction of two test functions
$f: V \to \R$ and $g: V \to \R$ and we use the Rayleigh quotient of $f$ to bound
$\lambda_2$ and the Rayleigh quotient of $g$ to bound $\lambda_n$. In our
construction, we have $|f(v)| = |g(v)|$ for all $v$ and $f(v)\ge 0$, and so $||f||^2 = ||g||^2$
and  the ratio of their Rayleigh quotients is simply 

\[ \frac {g^T Lg}{f^T L f}
= 1 + \frac {f^TAf - g^TAg}{f^TD f - f^TAf} \geq 1 + \frac {f^TAf - g^TAg}{f^TD
f} =   1 + \frac {f^TAf }{f^T D f}- \frac {g^TAg}{g^TD g}  \ , \]
where we use the fact that, for our definition of $f$ and $g$, we have  $f^T D f = g^T D g$.
In the standard proof of Alon-Boppana, one picks a start vertex $r$ and a cutoff
parameter $k$, and then one defines the test function $f$ such that $f(v) =
(d-1)^{-\ell/2}$, where $\ell$ is the distance from $r$ to $v$, for all vertices
$v$ at distance $\leq k$ from $r$; we set $f(v) = 0$ for vertices at distance
more than $k$ from $r$. In the analysis, one notes that vertices $v$ at distance
between $1$ and $k$ from $r$ contribute $2\sqrt {d-1} f^2(v)$ to the quadratic
form $f^T A f$ and contribute $f^2(v)$ to $||f||^2$, which is how one argues
that the $f^T A f / ||f||^2$ is at least about $2\sqrt {d-1}$.

In our construction, we pick a parameter $k$ smaller than the girth, we pick an
initial vertex $r$ at random, and we also define $f$ (and $g$) so that only
vertices at distance $\leq k$ from $r$ are nonzero in $f(\cdot)$. If $v$ is at
distance $\ell\leq k$ from $r$, the standard Alon-Boppana proof defines  $f(v)$
as being the square root of the probability of reaching $v$ in $\ell$ steps in a
non-backtracking random walk started at $r$ provided that, as in our case, one
assumes that the girth of the graph is more than $k$. Our definition in the
weighted case is similar but simpler to work with: we normalize weights so that
the maximum weighted degree is $1$, and we define $f(v)$ as the square root of
the product of the weights in the unique shortest path  from $r$ to $v$. Using
the facts that $f^2(v)$ is close to the probability of going from $r$ to $v$ in
a standard random walk (in which edges are picked proportionally to their weight),
that such a random walk is likely to be non-backtracking,  and that the random
walk in $G$ has a stationary distribution that is close to uniform, we relate
the contribution of an edge $(u,v)$ to $f^T A f$, averaged over random $r$, to
the average of $w(u,v)^{3/2}$ over all edges in the graph.  Finally, a convexity
argument shows that, up to lower order terms, this average is at least about $2
||f ||^2 / \sqrt d$ given that there are only $dn/2$ edges. A similar argument
applies to the construction of $g$, showing that one can have $g^T A g$  be at
most $-2 ||g||^2 / \sqrt d$ up to lower order terms. Finally one notes that $f^T
D f = g^T D g$ and $||f||^2 = ||g||^2$ are approximately the same.

\subsubsection{A Lowerbound for the \OSS\ Problem}

Our Alon-Boppana result shows that the best possible approximation achievable by
spectral sparsifiers with $dn/2$ edges cannot be
better than $1+\frac{4}{\sqrt{d}}$ up to lower
order terms in $d$. On the other hand, the result of \cite{BSS12} shows the
existence of sparsifiers the same number of  edges and error
$1+\frac{4\sqrt{2}}{\sqrt{d}}+o_d(1)$. It is natural to ask what the right dependence is, and
whether the constant $4\sqrt{2}$ can be improved to $4$ in general, or vice
versa. In this section, we note that the BSS algorithm actually solves a
more general problem, which we call \OSS, and we 
show that the best possible constant for that problem is $4\sqrt{2}$. Thus any improvement
on the density of spectral sparsifiers, if at all possible, will have to come from
an approach that does not also solve the \OSS\ problem. 

The \OSS\ problem is defined as follows. The player is given parameters $m,n$
and a number of rounds $T=dn/2$ in advance, and in each round $t=1,\ldots,T$ presented with a
collection of vectors $v^{(t)}_1,\ldots,v^{(t)}_m\in\R^n$ which are {\em
isotropic}, meaning:
$$\sum_{i=1}^m \vt_i(\vt_i)^T = I_n,$$
but can otherwise be chosen adversarially, depending on past actions. At each
time $t$ the player must choose an index $i(t)$ and a scaling $s_t$. The goal is to minimize the
condition number of the sum:
$$ A_T:= \sum_{t\le T} s_t \vt_{i(t)}(\vt_{i(t)})^T.$$

Although the theorem of \cite{BSS12} is stated for a fixed (static) set of
vectors, it is easy to see that the analysis
of the BSS algorithm allows one to change the set of vectors adversarially 
in every iteration, and an immediate consequence of the proof is the following.
\begin{theorem}\cite{BSS12} \label{th.onlinebss} There is a polynomial time online strategy which
solves the \OSS\ problem with $dn/2$ rounds with condition number at most 
$$\kappa_d := \frac{(\sqrt{d/2}+1)^2}{(\sqrt{d/2}-1)^2}=1+\frac{4\sqrt{2}}{\sqrt{d}}+O(1/d).$$
\end{theorem}
The corresponding result for spectrally sparsifying graphs $G$ follows by applying this strategy
to the fixed set of vectors $\{L_G^+(e_i-e_j)\}_{ij\in E}$.

Our second contribution is to show that the BSS algorithm is optimal for this
more general problem.
\begin{theorem}\label{th.online} There is no strategy for \OSS\ with $dn/2$
rounds which achieves condition number better than $\kappa_d-o_n(1)$.\end{theorem}

The conceptual point of this theorem is that achieving the true ``Ramanujan''
type bound of $1+\frac{4}{\sqrt{d}}$ will require an algorithm/analysis
which exploits one or both of the following facts: (1) the vectors are static
(2) the vectors have special structure, namely, they are (scaled) incidence
vectors of edges in a graph. It is conceivable that the online vector problem,
the offline vector problem, and the spectral graph sparsification problem are
all equally hard, or that each is strictly harder than the next. 

\section{Preliminaries}

Let $G= (V,E)$ be a weighted undirected graph, and $w(u,v)$ be the weight of edge $\{u,v\}$.
We refer to the {\em distance} between two vertices as the minimum number of edges in a path between them (that is, their unweighted shortest path distance).
The weighted degree of $u$ is defined as $w(u):= \sum_v w(u,v)$. The combinatorial degree
of $u$ is the number of edges incident on $u$ of nonzero weight.
If $W$ is the weighted adjacency matrix of $G$ (that is, $W_{u,v} = w(u,v)$) and $D$ is the 
diagonal matrix such that $D_{v,v}=w(v)$ is the weighted degree of $v$, then $L : = D-A$ is the
Laplacian matrix of $G$.

We identify vectors in $\R^V$ with functions $V \to \R$. The quadratic form of $L$ is

\[ f^T L f = \sum_{ \{ u,v \} } w(u,v)  \cdot ( f(u) - f(v))^2 \]

If we let $\lambda_1 \leq \lambda_2 \leq \cdots \lambda_n$ be the eigenvalues of $L$,
counted with multiplicities and ordered non-decreasingly, then 

\[ \lambda_2 = \min_{f \perp \bone} \  \ \frac{ f^T L f }{||f||^2} \]
\[ \lambda_n = \max_f \  \ \frac{ f^T L f }{||f||^2} \]

Without loss of generality, we may assume that the maximum weighted degree of $G$ is 1, because multiplying all edge weights by the same constant does not change the ratio $\lambda_n / \lambda_2$.

Next we observe that, without loss of generality, every node of $G$ has combinatorial degree $\geq \frac d4$, that the minimum weighted degree is at least $1-4/\sqrt d$ times the maximum weighted degree, and that every edge has weight at most $4/\sqrt d$. 

\begin{claim} Suppose that $G$ has a node of combinatorial degree $< d/4$. Then
\[ \frac {\lambda_n}{\lambda_2} \geq 1 + \frac 4{\sqrt d} - O \left( \frac {\sqrt d}{n} \right) \]
\end{claim}

\begin{proof}
This is proved in \cite{BSS12}. 
\end{proof}

\begin{claim} Suppose that $G$ has a node of weighted degree $\leq 1 - 4/\sqrt d$. Then
\[ \frac {\lambda_n}{\lambda_2} \geq 1 + \frac 4{\sqrt d} - O \left( \frac 1n \right) \]
\end{claim}

\begin{proof} Let $u$ be a node of weighted degree $\leq 1- \epsilon$ and let $v$ a node of weighted degree $1$.

Define the function $f: V \to \R$ such that $f(u) = 1$ and $f(z) = -1/(n-1)$ for $z\neq u$. Then $f \perp \bone$ and 
\[ \lambda_2 \leq \frac {f^T L f}{||f||^2} \leq \frac 1{||f||^2} \cdot \left( 1- \frac 4 {\sqrt d} \right) \left( 1 + \frac 1 {n-1} \right) \leq 1 -  \frac 4{\sqrt d} + O \left( \frac 1n \right) \]
Then define $h: V \to \R$ such that $h(v) = 1$ and $h(z) = 0$ for $z \neq v$ and observe that
\[ \lambda_n \geq \frac {h^T L h}{||h||^2} = 1 \]
So
\[ \frac {\lambda_n}{\lambda_2} \geq \frac 1 {1 -  \frac 4{\sqrt d} + O \left( \frac 1n \right)} \geq
1 +  \frac 4{\sqrt d} + O \left( \frac 1n \right) \]
\end{proof}

Note also that the above proof establishes
\[ \lambda_2 \leq 1 + O \left( \frac 1n \right) \]
which can also be verified by noting that the trace is at most $n$ and so $\lambda_2$ is at most $n/(n-1)$.

\begin{claim} Suppose that $G$ has an edge $\{ u,v \}$ of weight $> 4/\sqrt d$. Then
\[ \frac {\lambda_n}{\lambda_2} \geq 1 + \frac 4{\sqrt d} - O \left( \frac 1n \right) \]
\end{claim}

\begin{proof}
Let $h$ be such that $h(u) = 1$, $h(v) = -1$ and $h(z) = 0$ for $z \not\in \{ u,v \}$. Then
\[ \lambda_n \geq \frac { h^T L h}{|| h||^2} \geq \frac 1 {||h||^2} \cdot \left( 2\left( 1 - \frac 4{\sqrt d} \right) + 4 \cdot \frac {4}{\sqrt d} \right) \geq 1 + \frac 4 {\sqrt d} \]
\end{proof}

The girth of $G$ is the length (number of edges) of the shortest simple cycle in $G$.

We now notice that a girth assumption, combined with a lower bound on mimum combinatorial degree, implies an upper bound to the number of vertices in small balls.

\begin{claim} \label{claim.ball} Suppose $G$ has minimum combinatorial degree $\geq \frac d4$, that $d\geq 12$,  and that the girth of $G$ is at least $g$. Then, for every vertex $r$, and for every $\ell \leq (g-1)/2$, the number of vertices having distance $\leq \ell$ from $r$
is at most 
\[ \frac {2n}{\left( \frac d4 -1 \right) ^{\frac {g-1} 2 - \ell }} \]
\end{claim}

\begin{proof} It will be enough to show that the number of vertices at distance exactly $\ell$
is at most $n \cdot \left( \frac d4 -1 \right ) ^{\ell - \frac {g-1} 2}$. Let $s(r,i)$ be the number of vertices
at distance exactly $i$ from $r$. Then, for every $i < (g-1)/2$, we have $s(r,i) \geq \left( \frac d4 -1 \right) \cdot s(r,i-1)$, because the set of nodes at distance $< (g-1)/2$ from $r$ induces a tree in which all the non-leaf vertices have combinatorial degree $\geq d/4$. But $s(r,(g-1)/2) \leq n$, and
so 
\[ s(r,\ell) \cdot \left( \frac d4 - 1 \right) ^{\frac {g-1} 2 - \ell } \leq n \cdot s(r,(g-1)/2) \] 
\end{proof}

\section{Proof of Theorem \protect\ref{th.ab}}

Let $k$ be a parameter smaller than $(g-1)/2$, where $g$ is the girth,  to be set later (looking ahead,
we will set $k$ to be $d^{1/8}$). 

For every vertex $r$, let
$f_r:V\rightarrow\R$ be the function supported on the ball of radius $k$
centered at $r$ defined as follows:
\begin{align*}
	f_r(v) &=\begin{cases}  &0 \quad\textrm{if $\mathrm{dist}(r,v)>k$}\\
 &1 \quad\textrm{if $r=v$}\\
\\ &\sqrt{w(r,v_1)w(v_1,v_2)\ldots w(v_{\ell-1},v)}\quad\\&\textrm{otherwise, where
$r,v_1,\ldots,v_{\ell-1},v$ is the unique path of length $\ell\le k$ from $r$ to $v$}\\
	\end{cases}
\end{align*}

We begin by proving the following facts about $f_r$, which hold for every $r\in V$ and
which we will use repeatedly.

\begin{equation} \label{eq.fnorm}
\left( {1-\frac{8}{\sqrt{d}}} \right)^k \cdot (k+1) \le \|f_r\|_2^2 \le k+1
\end{equation}

\begin{equation} \label{eq.proj}
\|f_v^\perp\|_2^2 \geq \|f_r\|_2^2 \cdot \left( 1 - O \left( \frac 1{d^2}  \right ) \right) 
\end{equation}
where $f^\perp$ denotes
the projection of $f$ on the space orthogonal to the all ones vector.

To prove \eqref{eq.fnorm}, call $S(r,\ell)$ the set of nodes at distance exactly $\ell$
from $r$, and $C_\ell := \sum_{v \in S(r,\ell)} f_r(v)^2$ the contribution to $||f_r||^2$ of the nodes in $S(r,\ell)$. Then we have
\[ || f||_2^2 = \sum_{\ell = 0}^k C_\ell \]
and $C_0 = 1$, so it suffices to prove that, for $0 \leq \ell \leq k-1$, we have
\[ C_\ell \cdot \left( 1- \frac 8{\sqrt d} \right) \leq C_{\ell+1} \leq C_\ell \]
which follows from
\[ C_{\ell+1} = \sum_{v\in S(r,\ell+1)} f^2_{\parent(v)} w(u,v) = \sum_{u\in
S(r,\ell)} f^2_u \cdot ( w(u)- w(\parent(u) ,u)) \]
(there is an abuse of notation in the last expression: when $u=r$, then take
$w(\parent(r),r)$ to be zero) and from
the fact that $1-4/\sqrt d \leq w(u) \leq 1$ for every $u$, and the fact that all edges have weight at most $4/\sqrt d$.

To prove \eqref{eq.proj}, we see that

\[ || f_r^\perp ||_2^2 = || f_r ||^2 - ||f_r^{1} ||_2^2 \]
where $f^{1}_r$ is the projection of $f$ on the direction parallel to the all-one vector $\bone = (1,\ldots, 1)$, and
\[ || f^1_r||^2_2 = \left\langle f_r, \frac 1 {\sqrt n} \bone \right\rangle^2  = \frac 1n \left(\sum_v f_r \right)^2 = \frac 1n || f_r||_1^2 \leq \frac 1n || f_r||^2_2 \cdot ||f_r||^2_0 
\leq O \left ( \frac 1 {d^2} \right) \cdot  || f_r||^2_2 \]
where we used Claim \ref{claim.ball} to bound the ball of radius $k$ around $r$, which 
is the number of non-zero coordinates in $f$.

We now come to the core of the analysis

\begin{lemma}There exists a vertex $r$ such that $f_r^T W f_r\ge
2k/\sqrt{d}-O\left(k^2/d^{3/4}\right).$
\end{lemma}
\begin{proof}
    For any $r$, let $T_r$ denote the tree rooted at $r$
    of depth $k$ in $G$. We will think of the edges of $T_r$ as being directed edges $(u,v)$
    where $u$ is the parent of $v$. With some abuse of notation, we will also use $T_r$  to denote
    the set of vertices of $T_r$ and to denote the set of edges of $T_r$. 

    Recall by the definition of $f_r$ that if $u$ is the parent of $v$, then $f_r(v) = \sqrt{w(u,v)} f_r(u)$. 
    We have:

    \begin{align*}
    f_r^T W f_r & = 2\sum_{(u,v) \in T_r} w(u,v) f_r(u) f_r(v) \\
         & = 2\sum_{(u,v) \in T_r}  \sqrt{w(u,v)} f^2_r (v)  \\ 
         & = 2 \sum_{v \in T_r - \{ r \} }  \sqrt{w(\parent(v),v)} f^2_r (v)
     \end{align*}

	Consider now the simple random walk on $G$, where edges are selected with
	probability proportional to their weight. Then the transition probability from a
	vertex $u$ to a vertex $v$ is 
	$$p(u,v)=\frac{w(u,v)}{w(u)}.$$
	Recalling our assumption on the minimum weighted degree we have,
	$$ w(u, v)\ge \left( 1 - \frac 4 {\sqrt d} \right)  p(u,v)$$

	Let $\P_r$ denote the law of the $k-$step random walk
	$r=X_0,X_1,X_2,\ldots,X_k$ started at $r$. Suppose $v$ is a vertex in
$T_r$ at distance $\ell=\dist(r,v)$ from $r$ 
	and let $(r,v_1,\ldots,v_{\ell -1} , v)$ be the unique path from $r$ to
	$v$ in $T_r$ (and also in $G$, by the girth assumption). Then we have
$$f_r(v)^2  \geq \left( 1 - \frac 4{ \sqrt d} \right)^\ell
  p(r,v_1) \cdot p(v_1,v_2) \cdots p(v_{\ell-1} , v) \\
=  \left( 1 - \frac 4{ \sqrt d} \right)^\ell
\P_r\{X_{\dist(r,v)}=v\},$$
since traversing this path is the only way to reach $v$ in $\dist(r,v)$ steps.
Thus, we have for every choice of root $r\in V$:
	\begin{align*}
	 f_r^TWf_r &\ge 
	2\left( 1 - \frac 4{ \sqrt d} \right)^k \E_r \sum_{v\in T_r\setminus\{r\}} \{X_{\dist(r,v)}=v\} \sqrt{w(\mathrm{parent}(v),v)} 	
	\\ &= 
	2\left( 1 - \frac 4{ \sqrt d} \right)^k \E_r
	\sum_{i=1}^k \{\dist(r,X_i)=i\} \sqrt{w(\mathrm{parent}(X_i),X_i)} 	
	\\&\quad\textrm{since the walk can be at only one vertex at every step}	
	\\ &=
	2\left( 1 - \frac 4{ \sqrt d} \right)^k \E_r
	\sum_{i=1}^k \{\textrm{the walk is nonbacktracking up to step $i$}\}
\sqrt{w(X_{i-1},X_i)} 	
	\\ &\ge
	2\left( 1 - \frac 4{ \sqrt d} \right)^k \E_r\left[\{\textrm{the walk is
	nonbacktracking up to step $k$}\}\cdot \sum_{i=1}^k \sqrt{w(X_{i-1},X_i)}\right] 	
	\\ &=
	2\left( 1 - \frac 4{ \sqrt d} \right)^k\left( \E_r \sum_{i=1}^k
\sqrt{w(X_{i-1},X_i)} -\E_r \left[\{\textrm{the walk backtracks}\}\cdot
\sum_{i=1}^k \sqrt{w(X_{i-1},X_i)}\right] \right)
	\end{align*}

We will show that a good $r$ exists by averaging this bound over all $r$
according to the stationary distribution of the simple random walk: 
$$\pi(r)=\frac{w(r)}{\sum_{v\in V}w(v)}.$$ 
This will require a lowerbound on the first term and an upperbound on
the second term above, averaged over $r$. We achieve this in the following two
propositions, where $\P$ denotes the law of a
stationary $k-$step walk $\pi\sim X_0,X_1,\ldots,X_k$, and we have the relation
$$ \E(\cdot) = \sum_{r\in V} \pi(r)\E_r(\cdot).$$

\begin{prop} $$\E \sum_{i=1}^k \sqrt{w(X_{i-1},X_i)} \ge
\frac{k}{\sqrt{d}}-\frac{2k}{d}. $$\end{prop}
\begin{proof} Recall that the marginal distribution of every edge in a stationary random
walk is the same, and the edge $uv$ appears with probability proportional to
$w(u,v)$. Thus we have:
$$\E \sum_{i=1}^k \sqrt{w(X_{i-1},X_i)} = k\E\sqrt{w(X_0,X_1)}=k\cdot \frac{\sum_{uv \in E} w(u,v)^{3/2}}{\sum_{uv\in E} w(u,v)}.$$
	Since the function $x^{3/2}$ is convex the latter expression is
	minimized when all the $w(u,v)$ are equal; noting that $|E|=dn/2$, and
	$$S:=\sum_{uv\in E}
	w(u,v) = \frac12 \sum_{v\in V} w(v)\ge \left(1-\frac{4}{\sqrt{d}}\right)\cdot n/2$$
	 we have a lower bound of $k$ times
	$$ \frac{(dn/2)\cdot (S/(dn/2))^{3/2}}{S}\ge \sqrt{\frac{S}{dn/2}}\ge
\left(1-\frac4{\sqrt{d}}\right)^{1/2}\cdot\frac1{\sqrt{d}}=\frac{1}{\sqrt{d}}-\frac{2}{d}.$$

\end{proof}

\begin{prop} $$ \E \left[\{\textrm{the walk
	backtracks}\}\cdot \sum_{i=1}^k
\sqrt{w(X_{i-1},X_i)}\right]  \le \frac{40k^2}{d^{3/4}},$$
whenever $d\ge 25$. \end{prop}
\begin{proof}
Since every edge can be assumed to have weight at most $4/\sqrt{d}$, we have the
deterministic bound
$$ \sum_{i=1}^k \sqrt{(w(X_{i-1},X_i)} \le 2k/d^{1/4}.$$
Let $B_i$ denote the event that the walk backtracks at step $i$. Then we have
\begin{align*}
\P(B_2\lor \ldots \lor B_k) &\le \sum_{i=2}^k \P(B_i)
\\&\le (k-1)\cdot \frac{4/\sqrt{d}}{1-4/\sqrt{d}}
\\&\quad\textrm{since $p(u,v)\le w(u,v)/w(u)$ for every edge $(u,v)$}.
\\&\le \frac{20k}{\sqrt{d}}
\end{align*}
when $d\ge 25$. Combining this with the previous bound gives the desired result.
\end{proof}
Combining the above bounds gives:
\begin{align*}
\sum_{r\in V}\pi(r)f_r^TWf_r \ge
2\left(1-\frac4{\sqrt{d}}\right)^k\left(\frac{k}{\sqrt{d}}-\frac{2k}{d}-\frac{40k^2}{d^{3/4}}\right)
\ge \frac{2k}{\sqrt{d}}-O\left(\frac{k^2}{d^{3/4}}\right).
\end{align*}
Thus, there must exist a vertex $r$ satisfying the desired bound.
\end{proof}

Given the Lemma, the main result is obtained easily as follows.
\begin{proof}
	Let $f:=f_r$ from the previous Lemma and let $f'$ be $f$ with signs alternating
	at each level of the tree $T_r$. Observe that $f^TWf = -f'^TWf'$ since all
	edges are between levels of the tree. Thus, we have
	$$ f^T(D-A)f = f^TDf - f^TWf \le f^TDf - 2k(1-\delta)/\sqrt{d}$$
	and
	$$ f'^T(D-A)f' \ge f^TDf + 2k(1-\delta)/\sqrt{d}$$
	for some $\delta=O(k/d^{1/4})$, 
	since $f'^TDf'=f^TDf$. Thus, the ratio of these
	quantities is at least:
\begin{equation}\label{eq.ratio}\frac{f'^T(D-A)f}{f^T(D-A)f}=\frac{f^TDf +
2k(1-\delta)/\sqrt{d}}{f^TDf -
2k(1-\delta)/\sqrt{d}} \ge
1+\frac{4k(1-\delta)}{f^TDf\cdot\sqrt{d}}\ge
1+\frac{4k(1-\delta)}{(k+1)\sqrt{d}},\end{equation}
since $f^TDf\le \|f\|_2^2\le (k+1)$ by \eqref{eq.fnorm}. 

We now take $f^+$ and $f^-$ to be the projections of $f$ and $f'$
	orthogonal to the all ones vector; since the quadratic form of $L=D-A$
	is translation invariant this does not change the above quantities. The
ratio for the normalized vectors is now:
$$\frac{(f^-)^T Lf^-/\|f^-\|_2^2}{(f^+)^TLf^+/\|f^+\|_2^2} =
\frac{f'^T(D-A)f'}{f^T(D-A)f}\frac{\|f^\perp\|_2^2}{\|f'^\perp\|_2^2} \ge
\left(1+\frac{4}{\sqrt{d}}(1-\delta)(1-1/k)\right)\left(1-O(1/d^2)\right),$$
by \eqref{eq.ratio} and \eqref{eq.proj}. Setting $k=d^{1/8}$ gives the desired
bound.

\end{proof}

\section{Proof of Theorem \protect\ref{th.online}}
To ease notation and to be consistent with the proof in \cite{BSS12}, we will
let $\beta=d/2$ and talk about choosing $T=\beta n$ vectors instead of $dn/2$
vectors.
Let $n$ be a power of $4$ and
let $m=n$. Suppose $H_n$ is the Hadamard matrix of size $n$, normalized so
$\|H_n\|=1$, and let $h_1,\ldots,h_n$ be its columns. During any execution of the
game, let $$A_\tau:=\sum_{t\le \tau} s_t \vt_{i(t)} (\vt_{i(t)})^T$$ denote
the matrix obtained after $\tau$ rounds, with $A_0=0$. Consider the following
adversary:
\begin{quote}
In round $\tau+1$ present the player with vectors $v^{(\tau+1)}_1:=Uh_1,\ldots,
v^{(\tau+1)}_n:=Uh_n$, where $U$ is an orthogonal matrix whose columns form an
eigenbasis of $A_\tau$. 
\end{quote}

\noindent Note that the vectors $v^{(\tau+1)}_1,\ldots,v^{(\tau+1)}_n$ are
always isotropic since
$$\sum_{i=1}^n (Uh_i)(Uh_i)^T = UHH^TU^T = I.$$
We will show that playing any strategy against this adversary must incur a condition
number of at least 
$$\kappa_{d}-o_n(1)=
\frac{(\sqrt{\beta}+1)^2}{(\sqrt{\beta}-1)^2}=1+\frac{4}{\sqrt{\beta}}+O(1/\beta).$$
 Let
$p_\tau(x):=\det(xI-A_\tau)=\prod_{j=1}^n(x-\lambda_j)$ denote the
characteristic polynomial of $A_\tau$.  Observe that for any choice
$s=s_{\tau+1}$ and $v=Uh_i$ made by the player in round
$\tau+1$, we have:
\begin{align*}
p_{\tau+1}(x)&=\det(xI-A_\tau-svv^T)
\\&= \det(xI-A_\tau)\det(I-(xI-A)^{-1}(svv^T))
\\&=p_\tau(x)\left(1-s\sum_{j=1}^n \frac{\langle
v,u_j\rangle^2}{x-\lambda_j}\right)
\\&=p_\tau(x)\left(1-\frac{s}{n}\sum_{i=j}^n\frac1{x-\lambda_j}\right),
\\&\qquad\textrm{since $\langle Uh_i,u_j\rangle=\langle h_i,U^Tu_j\rangle=\langle
h_i, e_j\rangle=\pm 1$ for every $j$}
\\&= p_\tau(x)-(s/n)p'_\tau(x)
\\&= (1-(s/n)D)p_\tau(x),
\end{align*}
where $D$ denotes differentiation with respect to $x$. Thus, the characteristic
polynomial of $A_{\tau+1}$ does not depend on the choice of vector in round
$\tau+1$, but only on the scaling $s_{\tau+1}$. Applying this fact inductively
for all $T$ rounds, we have:
$$ p_T(x) = \prod_{t\le T} (1-(s_t/n)D) x^n,$$
since $p_0(x)=x^n$. 
Note that since every $p_\tau(x)$ is the characteristic polynomial of a symmetric
matrix, it must be real-rooted. 

\begin{remark} \label{rem.rr} Since the above calculation holds for all choices
of weights $s$ and matrices $A$, we have recovered the well-known fact that for
any real-rooted $p(x)$, the polynomial $(1-\alpha D)p(x)$ is also
real-rooted for real $\alpha$.\end{remark}

Let $S:=\sum_{t\le T} s_t/n$. We will show that among all assignments of the
weights $\{s_t\}$ with sum $S$, the roots of $p_T(x)$ are extremized when all of the
$s_t$ are equal, namely\footnote{To avoid confusion, we remark that in what follows $T$ is always a number and
never the transpose (we will be dealing only with polynomials, not matrices).}:
\begin{enumerate}
\item [(A)] $\lambda_{min}(p_T)\le \lambda_{min} (1-(S/T)D)^Tx^n.$
\item [(B)] $\lambda_{max}(p_T)\ge \lambda_{max} (1-(S/T)D)^Tx^n.$
\end{enumerate}
To do this, we will use some facts about majorization of roots of polynomials.
Recall that a nondecreasing sequence $b_1\le b_2\le\ldots b_n$ majorizes another
sequence $a_1\le\ldots\le a_n$ if $\sum_{j=1}^n a_j=\sum_{j=1}^n b_j$ and the partial sums satisfy:
$$ \sum_{j=1}^k a_j \ge \sum_{j=1}^k b_j$$
for $k=1,\ldots,n-1$. We will denote this by $(a_1,\ldots,a_n)\prec
(b_1,\ldots,b_n)$, and notice that this condition implies that $a_1\ge b_1$ and
$a_n\le b_n$, i.e., the extremal values of $a$ are more concentrated than those
of $b$. We will make use of the fact that for a given sum $S$, the uniform
sequence $(S/n,\ldots,S/n)$ is majorized by every other sequence with sum $S$.

We now appeal to the following theorem of Borcea and Br\"anden
\cite{borcea2010hyperbolicity}.
\begin{theorem}\label{th.bb} Suppose $L:\R_n[x]\rightarrow\R[x]$ is a linear
transformation on polynomials of degree $n$. If $L$ maps real-rooted polynomials to real-rooted
polynomials, then $L$ preserves majorization, i.e.
$$ \lambda(p)\prec \lambda(q) \quad\Rightarrow \lambda(L(p))\prec
\lambda(L(q)),$$
where $\lambda(p)$ is the vector of nondecreasing zeros of $p$.\end{theorem}

Let $$\phi(x):=(x-(S/T))^T$$ and let $\psi_T(x):=\prod_{t=1}^T (x-s_t/n).$ Observe that
$(S/T,\ldots,S/T)=\lambda(\phi)\prec \lambda(\psi_T)$, since the sum of the roots
of $\psi_T$ is $S$.  Consider the linear transformation $L:\R_T[x]\rightarrow \R[x]$ defined by:
$$L(p) = D^np(1/D)x^n,$$
and observe that for any monic polynomial with roots $\alpha_t$:
$$L\left(\prod_{t=1}^T (x-\alpha_t)\right) = \prod_{t=1}^T (1-\alpha_t D) x^n.$$
By remark \ref{rem.rr}, $L(p)$ is real-rooted whenever $p$ is real-rooted, so
Theorem \ref{th.bb} applies. We conclude that the roots of $L(\psi_T)=p_T(x)$ majorize
the roots of $L(\phi)=(1-(S/T)D)^T x^n,$ so items (A) and (B) follow.

To finish the proof, we observe (as in \cite{MSSICM}, Section 3.2) that 
$$ (1-(S/T)D)^Tx^n = \mathcal{L}_n^{(T-n)}(n^2x/S)=:\mathcal{L}(x)$$
where the right hand side is a scaling of an {\em associated Laguerre polynomial}. The
asymptotic distribution of the roots of such polynomials is known, and converges
to the Marchenko-Pastur law from Random Matrix Theory as $n\rightarrow\infty$. In particular, Theorem
4.4 of \cite{dette1995some} tells us that 
$$\lambda_{min} \mathcal{L}(x)\rightarrow
\frac{S}{n}\left(1-\sqrt{\frac{n}{T}}\right)^2$$
and
$$\lambda_{max} \mathcal{L}(x)\rightarrow
\frac{S}{n}\left(1+\sqrt{\frac{n}{T}}\right)^2,$$
as $n\rightarrow\infty$ with $T=\beta n$. Thus, the condition number of of $A_T$ is
at least
$$
\frac{\lambda_{max}\mathcal{L}(x)}{\lambda_{min}\mathcal{L}(x)}=\kappa_d-o_n(1),$$
as desired.

\subsection*{Acknowledgments}
We would like to thank Alexandra Kolla for helpful conversations, as well as the
Simons Institute for the Theory of Computing, where this work was carried out.
\bibliography{alonboppana}

\end{document}